\documentclass[graybox]{svmult}
\usepackage{graphicx}
\usepackage[all]{xy}
\usepackage{helvet,courier,type1cm}
\usepackage{amssymb,amsmath}
\usepackage{fancyhdr}
\graphicspath{%
{/home/rvb/MyDocs/Talks/Information/}
{/home/roman/MyDocs/Talks/Information/}
{/home/rvb/MyDocs/Articles/Hamming/}
{/home/roman/MyDocs/Articles/Hamming/}
{/home/rvb/lisp/ga/data/}
{/home/roman/lisp/ga/data/}
{/home/USER/MyDocs/Articles/Hamming/}
{/home/USER/lisp/ga/data/}
}

\newcommand{\bR}{\mathbb{R}}

\newcommand{\bE}{\mathbb{E}}

\newcommand{\cP}{\mathcal{P}}

\begin{document}




\title*{Relation between the Kantorovich-Wasserstein metric and the Kullback-Leibler divergence}
\titlerunning{Relation between Kantorovich metric and KL-divergence}
\author{Roman V. Belavkin}
\authorrunning{Roman Belavkin}
\institute{Roman V. Belavkin \at Middlesex University, London NW4 4BT, UK,
\email{R.Belavkin@mdx.ac.uk}}

\maketitle



\abstract{%
We discuss a relation between the Kantorovich-Wasserstein (KW) metric and the Kullback-Leibler (KL) divergence.  The former is defined using the optimal transport problem (OTP) in the Kantorovich formulation.  The latter is used to define entropy and mutual information, which appear in variational problems to find optimal channel (OCP) from the rate distortion and the value of information theories.  We show that OTP is equivalent to OCP with one additional constraint fixing the output measure, and therefore OCP with constraints on the KL-divergence gives a lower bound on the KW-metric.  The dual formulation of OTP allows us to explore the relation between the KL-divergence and the KW-metric using decomposition of the former based on the law of cosines.  This way we show the link between two divergences using the variational and geometric principles.
}

\section{Introduction}
\label{sec:intro}

The study of the optimal transport problem (OTP), initiated by Gaspar Monge \cite{Monge81}, was advanced greatly when Leonid Kantorovich reformulated the problem in the language of probability theory \cite{Kantorovich42}.  Let $X$ and $Y$ be two measurable sets, and let $\cP(X)$ and $\cP(Y)$ be the sets of all probability measures on $X$ and $Y$ respectively, and let $\cP(X\times Y)$ be the set of all joint probability measures on $X\times Y$.  Let $c:X\times Y\to\bR$ be a non-negative measurable function, which we shall refer to as the \emph{cost function}.  Often one takes $X\equiv Y$ and $c(x,y)$ to be a metric.  We remind that when $X$ is a complete and separable metric space (or if it is a homeomorphic image of it), then all probability measures on $X$ are Radon (i.e. inner regular).

The expected cost with respect to probability measure $w\in\cP(X\times Y)$ is the corresponding integral:
\[
\bE_w\{c\}:=\int_{X\times Y} c(x,y)\,dw(x,y)
\]
It is often assumed that the cost function is such that the above integral is lower semicontonuous or closed functional of $w$ (i.e. the set $\{w:\bE_w\{c\}\leq\upsilon\}$ is closed for all $\upsilon\in\bR$).  In particular, this is the case when $c(w):=\bE_w\{c\}$ is a continuous linear functional.

Given two probability measures $q\in\cP(X)$ and $p\in\cP(Y)$, we denote by $\Gamma[q,p]$ the set of all joint probability measures $w\in\cP(X\times Y)$ such that their marginal measures are $\pi_X w=q$ and $\pi_Y w=p$:
\[
\Gamma[q,p]:=\{w\in\cP(X\times Y): \pi_Xw=q,\ \pi_Yw=p\}
\]
Kantorovich's formulation of OTP is to find optimal joint probability measure in $\Gamma[q,p]$ minimizing the expected cost $\bE_w\{c\}$.  The optimal joint probability measure $w\in\cP(X\times Y)$ (or the corresponding conditional probability measure $dw(y\mid x)$) is called the \emph{optimal transportation plan}.  The corresponding optimal value is often denoted
\begin{equation}
  K_c[p,q]:=\inf\left\{\bE_w\{c\}:  w\in\Gamma[q,p]\right\}
  \label{eq:k-distance}
\end{equation}
The non-negative value above allows one to compare probability measures, and when the cost function $c(x,y)$ is a metric on $X\equiv Y$, then $K_c[p,q]$ is a metric on the set $\cP(X)$ of all probability measures on $X$, and it is often called the \emph{Wasserstein metric} due to a paper by Dobrushin \cite{Dobrushin70,Vasershtein69}, even though it was introduced much earlier by Kantorovich \cite{Kantorovich42}.  It is known that the Kantorovich-Wasserstein (KW) metric (or related to it Kantorovich-Rubinstein metric) induces a topology equivalent to the weak topology on $\cP(X)$ \cite{Bogachev07}.

Another important functional used to compare probability measures is the Kullback-Leibler divergence \cite{Kullback-Leibler51}:
\begin{equation}
  D[p,q]:=\int_X\left[\ln\frac{dp(x)}{dq(x)}\right]\,dp(x)
  \label{eq:kl-divergence}
\end{equation}
where it is assumed that $p$ is absolutely continuous with respect to $q$ (otherwise the divergence can be infinite).  It is not a metric, because it does not satisfy the symmetry and the triangle axioms, but it is non-negative, $D[p,q]\geq0$, and $D[p,q]=0$ if and only if $p=q$.  The KL-divergence has a number of useful and sometimes unique to it properties (e.g. see \cite{Belavkin15:_gsi15} for an overview), and it plays an important role in physics and information theory, because entropy and Shannon's information are defined using the KL-divergence.

The main question that we discuss in this paper is whether these two, seemingly unrelated divergences have anything in common.  In the next section, we recall some definitions and properties of the KL-divergence.  Then we show that the optimal transport problem (OTP) has an implicit constraint, which allows us to relate OTP to variational problems of finding an optimal channel (OCP) that were studied in the rate distortion and the information value theories \cite{Shannon48,Stratonovich65}.  Using the fact that OCP has fewer constraints than OTP, we show that OCP defines a lower bound on the Kantorovich metric, and it depends on the KL-divergence.  Then we consider the dual formulation of the OTP and introduce an additional constraint, which allows us to define another lower bound on the Kantorovich metric.  We then show that the KL-divergence can be decomposed into a sum, one element of which is this lower bound on the Kantorovich metric.

\section{Entropy, Information and the Optimal Channel Problem}


Entropy and Shannon's mutual information are defined using the KL-divergence.  In particular, \emph{entropy} of probability measure $p\in\cP(X)$ relative to a reference measure $r$ is defined as follows:
\begin{align*}
  H[p/r]&:=-\int_X \left[\ln \frac{dp(x)}{dr(x)}\right]\,dp(x)\\
  &=\ln r(X)-D[p,r/r(X)]
\end{align*}
where the second line is written assuming that the reference measure is finite $r(X)<\infty$.  It shows that entropy is equivalent up to a constant $\ln r(X)$ to negative KL-divergence from a normalized reference measure.  The entropy is usually defined with respect to some Haar measure as a reference, such as the counting measure (i.e. $r(E)=|E|$ for $E\subseteq X$ or $dr(x)=1$).  We shall often write $H[p]$ instead of $H[p/r]$, if the choice of a reference measure is clear (e.g. $dr(x)=1$ or $dr(x)=dx$).  We shall also use notation $H_p(X)$ and $H_p(X\mid Y)$ to distinguish between the prior and conditional entropies.

Shannon's mutual information between two random variables $x\in X$ and $y\in Y$ is defined as the KL-divergence of a joint probability measure $w\in\cP(X\times Y)$ from a product $q\otimes p\in\cP(X\times Y)$ of the marginal measures $\pi_X w=q$ and $\pi_Y w=p$:
\begin{align*}
  I(X,Y)&:=D[w,q\otimes p]=\int_{X\times Y}\left[\ln\frac{dw(x,y)}{dq(x)\,dp(y)}\right]\,dw(x,y)\\
  &=H_q(X)-H_q(X\mid Y)=H_p(Y)-H_p(Y\mid X)
\end{align*}
The second line shows that mutual information can be represented by the differences of entropies and the corresponding conditional entropies (i.e. computed respectively using the marginal $dp(y)$ and conditional probability measures  $dp(y\mid x)$).  If both unconditional and conditional entropies are non-negative (this is always possible with a proper choice of a reference measure), then we have inequalities $H_q(X\mid Y)\leq H_q(X)$ and $H_p(Y\mid X)\leq H_p(Y)$, because their differences (i.e. mutual information $I(X,Y)$) is non-negative.  In this case, mutual information satisfies Shannon's inequality:
\[
0\leq I(X,Y)\leq\min[H_q(X),H_p(Y)]
\]
Thus, information is the amount by which the entropy is reduced, and entropy can be defined as the supremum of information or as self-information \cite{Belavkin15:_maxent}:
\[
\sup\{I(X,Y):\pi_X w=q\}=I(X,X)=H_q(X)
\]
Here, we assume that $H_q(X\mid X)=0$ for the entropy of elementary conditional probability measure $q(E\mid x)=\delta_x(E)$, $E\subseteq X$.  Let us now consider the following variational problem.

Given probability measure $q\in\cP(X)$ and cost function $c:X\times Y\to\bR$, find optimal joint probability measure $w=w(\cdot\mid x)\otimes q\in\cP(X\times Y)$ minimizing the expected cost $\bE_w\{c\}$ subject to the constraint on mutual information $I(X,Y)\leq\lambda$.  Because the marginal measure $\pi_X w=q$ is fixed, this problem is really to find an optimal conditional probability $dw(y\mid x)$, which we refer to as the \emph{optimal channel}.  We shall denote the corresponding optimal value as follows:
\begin{equation}
  R_c[q](\lambda):=\inf\left\{\bE_w\{c\}: I(X,Y)\leq\lambda,\ \pi_X w=q\right\}
  \label{eq:ocp}
\end{equation}
This problem was originally studied in the rate distortion theory \cite{Shannon48} and later in the value of information theory \cite{Stratonovich65}.  The value of Shannon's mutual information is defined simply as the difference:
\[
V(\lambda):=R_c[q](0)-R_c[q](\lambda)
\]
It represents the maximum gain (in terms of reducing the expected cost) that is possible due to obtaining $\lambda$ amount of mutual information.

Let us compare the optimal values~(\ref{eq:ocp}) and (\ref{eq:k-distance}) of the OCP and Kantorovich's OTP problems.  On one hand, the OCP problem has only one marginal constraint $\pi_X w=q$, while the OTP has two constraints $\pi_X w=q$ and $\pi_Y w=p$.  On the other hand, the OCP has an information constraint $I(X,Y)\leq\lambda$.  Notice, however, that because fixing marginal measures $q$ and $p$ also fixes the values of their entropies $H_q(X)$ and $H_p(Y)$, the OTP has information constraint implicitly, because mutual information is bounded above $I(X,Y)\leq\min[H_q(X),H_p(Y)]$ by the entropies.  Therefore, in reality the OTP differs from OCP only by one extra constraint --- fixing the output measure $\pi_Y w=p$.  Let us define the following extended version of OTP by introducing the information constraint explicitly:
\[
K_c[p,q](\lambda):=\inf\left\{\bE_w\{c\}: I(X,Y)\leq\lambda,\ \pi_X w=q,\ \pi_Y w=p\right\}
\]
For $\lambda=\min[H_q(X),H_p(Y)]$ one recovers the original value $K_c[p,q]$ defined in~(\ref{eq:k-distance}).  It is also clear that the following inequality holds for any $\lambda$:
\[
R_c[q](\lambda)\leq K_c[p,q](\lambda)
\]
In fact, the equality holds if and only if both problems have the same joint probability measure $w\in\cP(X\times Y)$ as their solution.

\begin{theorem}
  Let $w_{OCP}$ and $w_{OTP}\in\cP(X\times Y)$ be optimal solutions to OCP and OTP problems with the same information constraint $I(X,Y)\leq\lambda$.  Then $R_c[q](\lambda)=K_c[p,q](\lambda)$ if and only if $w_{OCP}=w_{OTP}\in\Gamma[p,q]$.
  \label{th:ocp=otp}
\end{theorem}

\begin{proof}
  Measure $w_{OCP}$ is a solution to OCP if and only if it is an element $w_{OCP}\in\partial D^\ast[-\beta c,q\otimes p]$ of subdifferential at function $u(x,y)=-\beta\,c(x,y)$ of a convex functional
  \[
  D^\ast[u,q\otimes p]=\ln\int_{X\times Y} e^{u(x,y)}\,dq(x)\,dp(y)
  \]
  which is the Legendre-Fenchel transform of the KL-divergence $D[w,q\otimes p]$ considered as a functional in the first variable (i.e. $w$).  This can be shown using the standard method of Lagrange multipliers (e.g. see \cite{Stratonovich75:_inf,Belavkin11:_optim}).  If there is another optimal measure $w_{OTP}$ achieving the same optimal value, then it also must be an element of the subdifferential $\partial D^\ast[-\beta c,q\otimes p]$, as well as any convex combination $(1-t)w_{OCP}+tw_{OTP}$, $t\in[0,1]$, because subdifferential is a convex set.  But this means that the KL-divergence $D[w,q\otimes p]$, the dual of $D^\ast[u,q\otimes p]$, is not strictly convex, which is false.\qed
\end{proof}

The optimal solution to OCP has the following general form
\begin{equation}
  dw_{OCP}(x,y)=dq(x)\,dp(y)\,e^{-\beta\,c(x,y)-\kappa(\beta,x)}
  \label{eq:exp-solution}
\end{equation}
where the exponent $\beta$, sometimes called the \emph{inverse temperature}, is the inverse of the Lagrange multiplier $\beta^{-1}$ defined from the information constraint by the equality $I(X,Y)=\lambda$.  In fact, one can show that $\beta^{-1}=dV(\lambda)/d\lambda$.  The normalizing function $\kappa(\beta,x)=\ln\int_Y e^{-\beta\,c(x,y)}\,dp(y)$ is in general non-constant, and the solution~(\ref{eq:exp-solution}) depends on the marginal measure $q\in\cP(X)$.  One can show, however, that if the cost function is translation invariant (i.e. $c(x+a,y+a)=c(x,y)$), then the function $dq(x)\,e^{-\kappa(\beta,x)}=e^{-\kappa_0(\beta)}$ does not depend on $x$, which gives a simplified expression:
\[
dw_{OCP}(x,y)=dp(y)\,e^{-\beta\,c(x,y)-\kappa_0(\beta)}
\]
The measure above does not depend on the input marginal measure $q\in\cP(X)$ explicitly, but only via its influence on the output measure $p\in\cP(Y)$.

The optimal channel $w_{OCP}$ may not coincide with the optimal transportation plan $w_{OTP}$.  Interestingly, from game-theoretic point of view the optimal channels should be preferred, because they achieve smaller expected costs.  If specific output measure $\pi_Yw=p$ is important, however, then optimal channel can potentially be useful in the analysis of the optimal transportation plan.

Finally, let us point out in this section that the KL-divergence $D[p,q]$ between the measures $p$, $q\in\cP(X)$ can be related to mutual information via \emph{cross-information}:
\begin{equation}
  D[w,q\otimes q]=\underbrace{D[w,q\otimes p]}_{I(X,Y)}+D[p,q]
  \label{eq:cross-inf}
\end{equation}
The term cross-information for the KL-divergence $D[w,q\otimes q]$ (notice the difference from mutual information $D[w,q\otimes p]$) was introduced in~\cite{Belavkin15:_maxent} by analogy with cross-entropy.  The expression~(\ref{eq:cross-inf}) is a special case of Pythagorean theorem for the KL-divergence.  As was shown in \cite{Belavkin:_gsi13}, a joint probability measure $w\in\cP(X\times Y)$ together with its marginals $\pi_Xw=q$ and $\pi_Yw=p$ defines a triangle $(w,q\otimes p,q\otimes q)$ in $\cP(X\times Y)$, which is always a right triangle (and the same holds for triangle $(w,q\otimes p,p\otimes p)$).  This means that the KL-divergence between marginal measures $q$ and $p$ can be expressed as the difference:
\[
D[p,q]=D[w,q\otimes q]-I(X,Y)
\]
Taking into account the constraint $I(X,Y)\leq\lambda$ and assuming that OCP and OTP have the same solution $w\in\cP(X\times Y)$ (i.e. $w\in\Gamma[p,q]$), we can relate the KW-metric and the KL-divergence in one expression:
\[
K_c[p,q](\lambda)=\inf\left\{\bE_w\{c\}: D[p,q]\geq D[w,q\otimes q]-\lambda\,,\ w\in\Gamma[p,q]\right\}
\]

\section{Dual Formulation of the Optimal Transport Problem}

Kantorovich's great achievement was dual formulation of the optimal transport problem way before the development of convex analysis and the duality theory.  Given a cost function $c:X\times Y\to\bR$ consider real functions $f:X\to\bR$ and $g:Y\to\bR$ satisfying the condition: $f(x)-g(y)\leq c(x,y)$ for all $(x,y)\in X\times Y$.  Then the dual formulation is the following maximization over all such functions:
\begin{equation}
  J_c[p,q]:=\sup\left\{\bE_p\{f\}-\bE_q\{g\}: f(x)-g(y)\leq c(x,y)\right\}
  \label{eq:dual}
\end{equation}
where we assumed $X\equiv Y$.  It is clear that the following inequality holds:
\[
J_c[p,q]\leq K_c[p,q]
\]
We shall attempt to use this dual formulation to find another relation between the KL-divergence and the KW-metric.  First, consider the following decomposition of the KL-divergence:
\begin{align}
  D[p,q]&=D[p,r]+D[r,q]-\int_X\ln\frac{dq(x)}{dr(x)}\,[dp(x)-dr(x)]\label{eq:kl-cosine}\\
  &=D[p,r]-D[q,r]-\int_X\ln\frac{dq(x)}{dr(x)}\,[dp(x)-dq(x)]\label{eq:kl-minus}
\end{align}
Equation~(\ref{eq:kl-cosine}) is the \emph{law of cosines} for the KL-divergence (e.g. see \cite{Belavkin:_gsi13}).  It can be proved either by second order Taylor expansion in the first argument or directly by substitution.  Equation~(\ref{eq:kl-minus}) can be proved by using the formula:
\[
D[q,r]+D[r,q]=\int_X\ln\frac{dq(x)}{dr(x)}\,[dq(x)-dr(x)]
\]

We now consider functions $f(x)-g(y)\leq c(x,y)$ satisfying additional constraints:
\begin{align*}
  \beta f(x)&=\nabla D[p,r]=\ln\frac{dp(x)}{dr(x)}\,,\qquad\beta\geq0\\
  \alpha g(x)&=\nabla D[q,r]=\ln\frac{dq(x)}{dr(x)}\,,\qquad\alpha\geq0
\end{align*}
Thus, $\beta\,f$ and $\alpha\,g$ are the gradients of divergences $D[p,r]$ and $D[q,r]$ respectively, and this means that probability measures $p,q\in\cP(X)$ have the following exponential representations:
\begin{align*}
  dp(x)&=e^{\beta\,f(x)-\kappa[\beta f]}\,dr(x)\\
  dq(x)&=e^{\alpha\,g(x)-\kappa[\alpha g]}\,dr(x)
\end{align*}
where $\kappa[(\cdot)]=\ln\int_X e^{(\cdot)}\,dr(x)$ is the normalizing constant (the value of the cumulant generating function).  One can show that
\begin{align*}
  \frac{d}{d\beta}\kappa[\beta\,f]=\bE_p\{f\}\,,&& D[p,r]=\beta\,\bE_p\{f\}-\kappa[\beta\,f]\\
  \frac{d}{d\alpha}\kappa[\alpha\,g]=\bE_q\{g\}\,,&& D[q,r]=\alpha\,\bE_q\{g\}-\kappa[\alpha\,g]
\end{align*}
Substituting these formulae into~(\ref{eq:kl-minus}) we obtain
\[
D[p,q]=\beta\bE_p\{f\}-\alpha\bE_q\{g\}-\left(\kappa[\beta f]-\kappa[\alpha g]\right)-\alpha\int_X g(x)\,[dp(x)-dq(x)]
\]
Let us define the following value:
\[
J_{c,\varepsilon}[p,q]:=\frac{1}{\varepsilon}\left[\beta\bE_p\{f\}-\alpha\bE_q\{g\}\right]
\]
where $\varepsilon=\inf\{\epsilon\geq0:\beta f(x)-\alpha g(y)\leq \epsilon\,c(x,y)\}$. The value above reminds the value $J_c[p,q]$ of the dual problem to OTP, defined in~(\ref{eq:dual}).  However, because we also require that functions $f$ and $g$ to satisfy additional constraints (the gradient conditions), we have the following inequality:
\[
J_{c,\varepsilon}[p,q]\leq J_c[p,q]\leq K_c[p,q]
\]
Using these inequalities, we can rewrite equation~(\ref{eq:kl-minus}) as follows:
\[
D[p,q]\leq\varepsilon\,K_c[p,q]-\left(\kappa[\beta f]-\kappa[\alpha g]\right)-\alpha\int g(x)\,[dp(x)-dq(x)]
\]

\begin{theorem}
  Let the pair of functions $(f,g)$ be the solution to the dual OTP~(\ref{eq:dual}).  If there exists a reference measure $r\in\cP(X)$ such that $f=\nabla D[p,r]$ and $g=\nabla D[q,r]$, then
  \[
  D[p,q]=K_c[p,q]-\left(\kappa[f]-\kappa[g]\right)-\int g(x)\,[dp(x)-dq(x)]
  \]
  \label{th:dual-otp=kl}
\end{theorem}

\begin{proof}
  The assumptions $f=\nabla D[p,r]$ and $g=\nabla D[q,r]$ mean that the Lagrange multipliers are $\alpha=\beta=1$, and probability measures have the form $p=\exp(f-\kappa[f])\,r$ and $q=\exp(g-\kappa[g])\,r$.  Substituting these expressions into equation~(\ref{eq:kl-minus}) will result in the expression containing the difference of expectations $\bE_p\{f\}-\bE_q\{g\}$, which equals to $J_c[p,q]=K_c[p,q]$.\qed
\end{proof}

\section*{Discussion}

In their original definitions, the optimal transport problem and the related to it Kantorovich-Wasserstein metric have no connection to the Kullback-Leibler divergence.  We have demonstrated that by relaxing one constraint, namely fixing the output measure, the optimal transport problem becomes mathematically equivalent to the optimal channel problem in information theory, which uses a constraint on the KL-divergence between the joint and the product of marginal measures (i.e. on mutual information).  This way, an optimal channel defines a lower bound on the KW-metric.  Interestingly, for this reason optimal channels should be preferred to optimal transportation plans purely from game-theoretic point of view.  Applying Pythagorean theorem for joint and product of marginal measures allowed us to relate the constraint on mutual information to the constraint on the KL-divergence between the marginal measures of optimal channel.

In addition to this variational approach, we have considered a geometric idea based on the law of cosines for the KL-divergence to decompose the divergence between two probability measures into a sum that includes divergences from a third reference measure.  We have shown then that a component of this decomposition can be related to the dual formulation of the optimal transport problem.

Generally, the relations presented have a form of inequalities.  Additional conditions have been derived in Theorems~\ref{th:ocp=otp} and \ref{th:dual-otp=kl} for the cases when the relations hold with equalities.

\begin{acknowledgement}
This work is dedicated to the anniversary of Professor Shun Ichi Amari, one of the founders of information geometry.  The work was supported in part by the Biotechnology and Biological Sciences Research Council [grant numbers BB/L009579/1, BB/M021106/1].
\end{acknowledgement}



\bibliography{rvb,nn,other,newbib,ica,evolution,transport}

\end{document}